\documentclass[12pt]{amsart}
\usepackage{amsmath}
\usepackage{amsthm}
\usepackage{amssymb}
\usepackage{amscd}
\usepackage{bbold}
\usepackage[pdftex]{graphicx}
\usepackage{enumitem}
\usepackage{bbm}
\usepackage[utf8]{inputenc}
\usepackage[english]{babel}
\usepackage{mathrsfs}

\usepackage{subcaption}

\newtheorem{lemma}{Lemma}[section]
\newtheorem{theorem}[lemma]{Theorem}
\newtheorem{proposition}[lemma]{Proposition}
\newtheorem{corollary}[lemma]{Corollary}

\theoremstyle{definition}
\newtheorem{definition}[lemma]{Definition}

\newtheorem{remark}[lemma]{Remark}

\newcommand{\CC}{{\mathbb C}}

\newcommand{\NN}{{\mathbb N}}

\newcommand{\RR}{{\mathbb R}}

\newcommand{\ZZ}{{\mathbb Z}}

\def\benm{\begin{enumerate}}            % Begin enumerate command
\def\eenm{\end{enumerate}}              % End enumerate command
\title{Three Dimensional Sums of Character Gabor Systems}

\date{}



\author{Kung-Ching Lin}
\address{Norbert Wiener Center\\
         Department of Mathematics \\
         University of Maryland \\
         College Park, MD 20742 \\
         USA}
\email{kclin@math.umd.edu}
\begin{document}

\newcounter{bean}

\begin{abstract}
	In deterministic compressive sensing, one constructs sampling matrices that recover sparse signals from highly incomplete measurements. However, the so-called square-root bottleneck limits the usefulness of such matrices, as they are only able to recover exceedingly sparse signals with respect to the matrix dimension. In view of the flat restricted isometry property (flat RIP) proposed by Bourgain et al., we provide a partial solution to the bottleneck problem with the Gabor system of Legendre symbols. When summing over consecutive vectors, the estimate gives a nontrivial upper bound required for the bottleneck problem.

\end{abstract}

\maketitle

%%%%%%%%%%%%%%%%%%%%%%%%%%%%%%%%%%%%%%%%%%%%%%%%%%%%%%%%%
%%%%%%%%%%%%%%%%%%%%%%%%%%%%%%%%%%%%%%%%%%%%%%%%%%%%%%%%%%

%%%Type manuscript

\section{Introduction and Motivation}

	In this paper we discuss the following sum: Given a prime $p\in\NN$ and $n\in\ZZ/p\ZZ$, suppose that $M_1, M_2\subset \ZZ/p\ZZ$ are two sets of consecutive numbers with $|M_1|\leq |M_2|\leq\sqrt{p}$, we would like to estimate

\begin{equation}
\label{inner_prod_sum}
	|\sum_{k}\sum_{m_1\in M_1}\sum_{m_2\in M_2}\chi[k+m_1-m_2]\chi[k]e^{2\pi\imath kn/p}e^{-2\pi\imath m_2n/p}|,
\end{equation}
	where $\chi:\ZZ/p\ZZ\to\CC$ is a non-principal character.

	The sum in $\eqref{inner_prod_sum}$ is related to deterministic compressive sensing, character sums, and Weil's exponential sum estimates. From all prior works, one can easily derive an upper bound of $p^{3/2}$ for \eqref{inner_prod_sum}. However, as such an estimate is not sufficient for our purpose, we shall prove that it is possible to improve the estimate to $p^{3/2-\alpha}$ under certain mild assumptions, where $\alpha\in(0,1/2)$ depends on $|M_1|$ and $n$.

\subsection{Deterministic Compressive Sensing and Flat Restricted Isometry Property}

	Introduced in \cite{EC_TT_2006} and refined in \cite{EC_TT_2005}, the Restricted Isometry Property (RIP) is defined as follows: 
\begin{definition}
	An $n\times m$ matrix $A$ satisfies $(S,\delta_S)$-RIP if the following statement is true: Let $A_T$, $T\subset\{1,\dots,m\}$ be the $n\times|T|$ submatrix obtained by extracting the columns of $A$ which corresponds to the elements in $T$. Then for any subset $T$ with $|T|\leq S$ and any coefficient sequence $\{c_j\}_{j\in T}$, we have
\begin{equation}
	(1-\delta_S)\|c\|_2^2\leq\|A_T c\|_{2}^2\leq (1+\delta_S)\|c\|_2^2.
\end{equation}
\end{definition}

	For sampling schemes satisfying RIP, one is able to retrieve sparse signals efficiently from highly incomplete measurements because of the equivalence between the following optimization problems:
\begin{equation}
\tag{$P_0$}
\label{P_0}
	\min \|x\|_{\ell_0}\quad\text{subject to } Ax=b,
\end{equation}
	where $\|x\|_{\ell_0}$ denotes the number of nonzero entries of $x$, and
\begin{equation}
\tag{$P_1$}
\label{P_1}
	\min \|x\|_{\ell_1}\quad\text{subject to } Ax=b.
\end{equation}
	\eqref{P_0} and \eqref{P_1} do not yield the same solution in general, but for matrices satisfying RIP with small constant $\delta$, the two problems will be equivalent provided that the signal itself is sparse, \cite{EC_JR_TT_2006_1}. \eqref{P_0} is a non-convex optimization problem, whereas \eqref{P_1} is convex and is readily solvable. Thus, solving \eqref{P_1} is much more preferable to solving \eqref{P_0}.
	
	Using probabilistic estimates, one can show that given $\epsilon>0$, there exists a random matrix $A\in\CC^{M\times N}$ satisfies $(S,\delta_S)$-RIP with $M^{1-\epsilon}\ll S\ll M$ with exponentially high probability. However, deterministically one is not able to obtain such strong results: Very few methods are available other than the coherence estimate, and it is extremely hard to extend the order $S$ to $S\gg\sqrt{M}$. Such difficulty is denoted as the square-root bottleneck.
	
	Bourgain et al.\ \cite{JB_2011} proposed a new class of matrices satisfying RIP of high order, breaking the bottleneck by constructing a family of matrices satisfying $(S,\delta_S)$-RIP with $S\sim M^{1/2+\epsilon}$, where $\epsilon$ is of the order of $10^{-28}$. Mixon \cite{DM_2015} improved the $\epsilon$ to the order of $10^{-24}$, more than $8,000$ times better than the original result. One key ingredient of their proofs is the following notion of flat RIP. 
	
\begin{definition}[flat RIP]
Let $u_1,\dots,u_N$ be the columns of an $n\times N$ matrix $\Phi$. Suppose that for every $j$, $\|u_j\|_2=1$. $\Phi$ satisfies the $(k,\delta)$-flat RIP if for any disjoint $J_1, J_2\subset \{1,\dots, N\}$ with $|J_1|, |J_2|\leq k$ we have
\begin{equation}
\label{fRIP}
	|<\sum_{j\in J_1}u_j, \sum_{i\in J_2}u_i>|\leq\delta(|J_1||J_2|)^{1/2}.
\end{equation} 
\end{definition}

	For the theory of deterministic compressive sensing, the coherence parameter $\mu$ of the given matrix is important:
\begin{definition}
	Given a matrix $\Phi=(\phi_1\mid\phi_2\mid\dots\mid\phi_r)$ with unit column vectors, the coherence parameter $\mu$ of $\Phi$ is defined to be
\[
	\mu:=\max_{i\neq j}|{<}\phi_i,\phi_j{>}|.	
\] 

\end{definition}
	The following lemma takes a slightly weaker form of flat RIP.
\begin{lemma}
\label{fRIP_RIP}
	Let $k\geq 2^{10}$ and $s$ be any positive integer. Assume that the coherence parameter of $\Phi$ is $\mu\leq1/k$, and for some $\delta$ and any disjoint $J_1,J_2$ with $|J_1|,|J_2|\leq k$, one has
\begin{equation}
\label{relaxed_fRIP}
	\bigg|{<}\sum_{j_1\in J_1}u_{j_1},\sum_{j_2\in J_2}u_{j_2}{>}\bigg|\leq\delta k,
\end{equation}
	then $\Phi$ satisfies RIP of order $(2sk, 44s\delta\log k)$-RIP.
\end{lemma}

	By Lemma \ref{fRIP_RIP}, matrices satisfying flat RIP also satisfy RIP of high order, which provides insights on how to approach this problem from a new direction.\par

	Motivated by this, we aim to construct deterministic matrices with bottleneck-breaking RIP from the Gabor system of Legendre symbols. Our formulation follows from \eqref{relaxed_fRIP}: given a prime $p\in\NN$, consider $\{u_{l,j}\}_{l,j\in\ZZ/p\ZZ}\subset\CC^p$ where $u_{l,j}[k]=\frac{1}{\sqrt{p}}\chi[k-l]e^{-2\pi\imath kj/p}$ with $\chi$ being the Legendre symbol. Fix disjoint $\Omega_1,\Omega_2\subset\ZZ/p\ZZ\times\ZZ/p\ZZ$ where $|\Omega_1|,|\Omega_2|\leq\sqrt{p}$, define $\pi_2(\Omega_i)=\{j\in\ZZ/p\ZZ: \exists l\in\ZZ/p\ZZ \text{ such that } (l,j)\in\Omega_i\}$ and $\Omega_i(j)=\{l\in\ZZ/p\ZZ: (l,j)\in M_i\}$ for $i=1,2$. Then,
\begin{equation}
\label{original_sum}
\begin{split}
	&\bigg|{<}\sum_{(m_1,n_1)\in\Omega_1}u_{m_1,n_1},\sum_{(m_2,n_2)\in\Omega_2}u_{m_2,n_2}{>}\bigg|\\
	&=\bigg|\frac{1}{p}\sum_{n_1\in\pi_2(\Omega_1)}\sum_{n_2\in\pi_2(\Omega_2)}\sum_{k\in\ZZ/p\ZZ}\sum_{m_1\in\Omega_1(n_1)}\sum_{m_2\in\Omega_2(n_2)}\chi[k+m_1-m_2]\chi[k]e^{2\pi\imath k(n_1-n_2)/p}e^{-2\pi\imath m_2(n_1-n_2)/p}\bigg|\\
	&\leq\frac{1}{p}\sum_{n_1\in\pi_2(\Omega_1)}\sum_{n_2\in\pi_2(\Omega_2)}\bigg|\sum_{k\in\ZZ/p\ZZ}\sum_{m_1\in\Omega_1(n_1)}\sum_{m_2\in\Omega_2(n_2)}\chi[k+m_1-m_2]\chi[k]e^{2\pi\imath k(n_1-n_2)/p}e^{-2\pi\imath m_2(n_1-n_2)/p}\bigg|.
\end{split}
\end{equation}
	Note that the expression in inside the final absolute value of \eqref{original_sum} is exactly \eqref{inner_prod_sum} when $\Omega_1(n_1),\Omega_2(n_2)$ are consecutive numbers. In order to use Lemma \ref{fRIP_RIP}, we aim to show that \eqref{inner_prod_sum} is less than $p^{3/2-\alpha}$ for some $\alpha>0$.

\subsection{Character Sum Estimates}	

	Besides the practical interests in compressive sensing, estimation of character sums is also intriguing in its own. Let $\chi:\ZZ/p\ZZ\to\CC$ be a non-principal character on $(\ZZ/p\ZZ)^\ast$ with the extension $\chi[0]=0$. Polya-Vinogradov inequality states that

\[
	|\sum_{M\leq k\leq M+N}\chi[k]|\leq\sqrt{p}\log p
\]

	for any arbitrary $M, N$. Chung \cite{FC_1994} investigated the cancellation within the sum

\[
	\sum_{a\in S}\sum_{b\in T}\chi[a+b]
\]
	where $S, T\subset\ZZ/p\ZZ$. In particular, the following estimate is given:
\[
	|\sum_{a\in S}\sum_{b\in T}\chi[a+b]|\leq\sqrt{p|S||T|}(1-\frac{|S|}{p})^{1/2}(1-\frac{|T|}{p})^{1/2}.
\]

	Note that the estimate is only nontrivial for $|S|, |T|\gg\sqrt{p}$. Chung also commented on a conjecture for the case $|S|\ll \sqrt{p}$: for any fixed $\epsilon>0$ and $|S|>p^\epsilon$, there exists $\delta>0$ such that
\[
	|\sum_{a,b\in S}\chi[a-b]|<|S|^{2-\delta}.
\]
	
	Friedlander and Iwaniec \cite{JF_HI_1993} gave a partial answer to the conjecture above, proving the inequality when $S$ is contained in an interval $I$ of length $\ll \sqrt{p}$ and satisfies $|S|\geq I^{r/(r+1)}p^{1/4r+\epsilon}$ for some $r\geq2$ using the Burgess estimate. Note that the results here do not apply to \eqref{inner_prod_sum} even if $\Omega_1(n_1)=\Omega_2(n_2)$, since there is an additional summation over $\ZZ/p\ZZ$.

\subsection{Weil's Exponential Sum Estimate}
	Using Weil's estimate, one has the following inequalities \cite{WS_2006, AB_MF_DM_JM_2016, AW_1948, JB_RB_JW_2012}:
\begin{theorem}
	Given a prime $p$ with $0<d_1<\dots<d_k<p$, one has
\[
	|\sum_{n=0}^{p-1}\chi[n+d_1]\cdots\chi[n+d_k]|\leq 9kp^{1/2}.	
\]
\end{theorem}

\begin{theorem}
	Given a prime $p$ and $m,n\in\ZZ/p\ZZ\backslash\{0\}$, one has
\[
	\bigg|\sum_{k\in\ZZ/p\ZZ}\chi[k]\chi[k+m]e^{-2\pi\imath kn/p}\bigg|\leq2\sqrt{p}.
\]
\end{theorem}
	
	In particular, the sum \eqref{inner_prod_sum} has the trivial estimate $\sqrt{p}|M_1||M_2|$. When $|M_1|, |M_2|\sim \sqrt{p}$, we will have that \eqref{inner_prod_sum}$\leq p^{3/2}$.

	In our case, the summation is three dimensional, complicating the issue. However, we shall show that if we add sufficiently large spins on the sum, there are indeed additional cancellations occurring.

%=====================%
\section{Main Results}
%=====================%

\begin{theorem}
\label{auto_corr_main}
	Let $p$ be a prime, and $n\in\ZZ/p\ZZ$. Suppose $n\sim p^{1/2+\delta}$, where $\delta\in(0,1/2)$, and $M_1, M_2\subset\ZZ/p\ZZ$ consist of consecutive numbers such that $|M_1|,|M_2|\leq\sqrt{p}$. Furthermore, if $|M_2|/|M_1|, |M_1|$ are even, and $|M_1|\sim p^{1/2-\sigma}$, $\sigma\in[0,1/2)$ such that $\delta>\sigma$, then
\begin{equation}
\label{sine_estimate}
	\sum_{s\neq 0, -n}|\frac{\sin(\pi|M_1|s/p)}{\sin(\pi s/p)}||\frac{\sin(\pi|M_2|(s+n)/p)}{\sin(\pi (s+n)/p)}|=O(p^{3/2-\alpha}),
\end{equation}
	where $\alpha=\sigma+(\delta-\sigma)/2$, and the big-O notation $A(p)=O(p^{3/2-\alpha})$ means that there exists a constant $K$, independent of $p$, such that $\limsup_{p:prime}\frac{A(p)}{p^{3/2-\alpha}}\leq K$.
\end{theorem}

	From this theorem, we derive the following corollaries:

\begin{corollary}
\label{auto_corr_dirichlet}
	With the assumptions above, we have
\begin{equation}
\label{auto_est}
	|\sum_{k}\sum_{m_1\in\Omega_1(n_1)}\sum_{m_2\in\Omega_2(n_2)}\chi[k+m_1-m_2]\chi[k]e^{2\pi\imath kn/p}e^{-2\pi\imath m_2n/p}|=O(p^{3/2-\alpha}),
\end{equation}
	where $n=n_1-n_2$.
\end{corollary}
	
\begin{corollary}
\label{inferior_one}
	With the same assumptions above, we have, for a fixed $k\in\ZZ/p\ZZ$,
\[
	|\sum_{m_1\in\Omega_1(n_1)}\sum_{m_2\in\Omega_2(n_2)}\chi[k+m_1-m_2]e^{2\pi\imath m_2n/p}|=O(p^{1-\alpha}).
\]
\end{corollary}

\begin{proof} of Corollary \ref{auto_corr_dirichlet}:

	Given $n\in\ZZ/p\ZZ$, we compute
\[
\begin{split}
	&\sum_{k}\sum_{m_1\in\Omega_1(n_1), m_2\in\Omega_2(n_2)}\chi[k+m_1-m_2]\chi[k]e^{2\pi\imath kn/p}e^{-2\pi\imath m_2n/p}\\
	&=\sum_{k,m_1,m_2}\bigg(\frac{1}{\sqrt{p}}\sum_s\chi[s]e^{2\pi\imath(k+m_1-m_2)s/p}\bigg)e^{-2\pi\imath m_2n/p}\chi[k]e^{2\pi\imath kn/p}\\
	&=\sum_s\chi[s]\bigg(\frac{1}{\sqrt{p}}\sum_k\chi[k]e^{2\pi\imath k(n+s)/p}\bigg)\bigg(\sum_{m_1}e^{2\pi\imath m_1s/p}\bigg)\bigg(\sum_{m_2}e^{-2\pi\imath m_2(s+n)/p}\bigg)\\
	&=\sum_s\chi[s]\chi[n+s]\bigg(\sum_{m_1}e^{2\pi\imath m_1s/p}\bigg)\bigg(\sum_{m_2}e^{-2\pi\imath m_2(s+n)/p}\bigg)\\
	&=\sum_{s\neq 0, -n}\chi[s]\chi[n+s]\bigg(\sum_{m_1}e^{2\pi\imath m_1s/p}\bigg)\bigg(\sum_{m_2}e^{-2\pi\imath m_2(s+n)/p}\bigg).
\end{split}
\]
	Assuming $\Omega_1(n_1), \Omega_2(n_2)$ are both intervals in $\ZZ/p\ZZ$, we see that 
\[
|\sum_{m_j\in \Omega_j(n_j)}e^{2\pi\imath m_j t/p}|=|\frac{\sin(\pi|M_j|t/p)}{\sin(\pi t/p)}|,
\]
	where $j=1,2$. Thus, taking the absolute value on both sides, we get this estimate.
\end{proof}

	The proof of Corollary \ref{inferior_one} follows verbatim.

\begin{remark}
	Using H\"{o}lder's inequality and the Fourier transform of the Fej\'{e}r's kernel, we can show that the expression in \eqref{sine_estimate} is less than $p\sqrt{|\Omega_1(n_1)||\Omega_2(n_2)|}$, which equals $p^{3/2}$ when $|\Omega_1(n_1)|=|\Omega_2(n_2)|=\sqrt{p}$.
\end{remark}

	To prove Theorem \ref{auto_corr_main}, we will approximate $\sin(\pi|M_j|(s+t_j)/p)$ and $\sin(\pi (s+t_j)/p)$ with piece-wise linear functions. Then, by summing over all pieces, we shall show that the contribution as a whole is less than $p^{3/2-\alpha}$. 	
\begin{definition}
	We define the following piece-wise polynomials $p_1^u, p_1^l, p_2^u, p_2^l$ as
\[
\left\{\begin{array}{ll}
	p_1^u(s)=2\||M_1|s/p\|,& p_1^l(s)=\| s/p\|,\\
	p_2^u(s)=2\||M_2|(s+n)/p\|,& p_2^l(s)=\| (s+n)/p\|,
\end{array}\right.
\]
	where $\|t\|:=\min_{n\in\ZZ}|t-n|$.
\end{definition}
	Note that
\[
	|\frac{\sin(\pi|M_1|s/p)}{\sin(\pi s/p)}||\frac{\sin(\pi|M_2|(s+n)/p)}{\sin(\pi (s+n)/p)}|\leq \frac{p_1^u(s)p_2^u(s)}{p_1^l(s)p_2^l(s)}.
\]

	As we assume that $|M_2|\geq|M_1|$, the piece-wise linear function of $|\sin(\pi|M_2|(s+n)/p)|$ changes directions most frequently. Thus, we first start with the intervals in which the function does not change direction before expanding into larger intervals. In particular, we define the following intervals:

\begin{definition}
	An interval in $\ZZ/p\ZZ$ with the form $[\frac{pj}{|M_2|}-n,\frac{p(j+1)}{|M_2|}-n]$, $j\in\{-|M_2|/2,\dots, |M_2|/2\}$ is called an $y_j$-interval, by which we denote $I^y_{j}$.
	
	An interval in $\ZZ/p\ZZ$ with the form $[\frac{pi}{|M_1|},\frac{p(i+1)}{|M_1|}]$, $i\in\{-|M_1|/2,\dots, |M_1|/2\}$ is called an $x_j$-interval, by which we denote $I^x_i$.
\end{definition}

	Here, we abuse the notation by denoting the set of numbers $\{a\in\ZZ/p\ZZ: a\in I\}\equiv I$ where $I\subset\RR$ is an interval.

	Given $s\in I_j^y\subset I_i^x$, we denote $x_i,y_j\in\ZZ$ by the integers such that $p_1^u(s)=|\frac{|M_1|s}{p}-x_i|$, $p_2^u(s)=|\frac{|M_2|(s+n)}{p}-y_j|$.

%======================%
\section{Proof of Theorem \ref{auto_corr_main}}
%======================%

	In this section, we track only the main terms occurring during the calculation. The error terms will be dealt with in Section \ref{error_term}.

	First, we see that
\[
	\bigg|\frac{\sin(\frac{\pi|M_1|s}{p})}{\sin(\frac{\pi s}{p})}\bigg|\cdot \bigg|\frac{\sin(\frac{\pi|M_2|(s+n)}{p})}{\sin(\frac{\pi(s+n)}{p})}\bigg|\leq\frac{4p^2}{\pi^2}\frac{|\frac{|M_1|s}{p}-x_i||\frac{|M_2|(s+n)}{p}-y_j|}{s(s+n)}=\frac{p_1^u(s)p_2^u(s)}{p_1^l(s)p_2^l(s)},
\]
	where $x_i=x_i(s)\in\{-\lceil\frac{|M_1|}{2}\rceil,-\lceil\frac{|M_1|}{2}\rceil+1,\dots, \lceil\frac{|M_1|}{2}\rceil\}\cap2\ZZ$, $y_j=y_j(s,n)\in\{\lfloor\frac{|M_2|n}{p}\rfloor,\dots, \lceil \frac{|M_2|}{2}+\frac{|M_2|n}{p}\rceil\}\cap2\ZZ$.

	Note that $I_j^y\subset I_i^x\iff y_j\in[x_i|M_2|/|M_1|+|M_2|n/p,x_{i+1}|M_2|/|M_1|+|M_2|n/p-1]=:J_i^x$. Then,
	
\[
	\sum_{s\neq 0, -n}|\frac{\sin(\pi|M_1|s/p)}{\sin(\pi s/p)}||\frac{\sin(\pi|M_2|(s+n)/p)}{\sin(\pi (s+n)/p)}|\leq\sum_{x_i=-|M_1|/2}^{|M_1|/2}\sum_{y_j\in J_i^x}\sum_{s\in I_j^y}\frac{p_1^u(s)p_2^u(s)}{p_1^l(s)p_2^l(s)}.
\]

	In our proof, we would like to smooth out $\{x_i\}_i,\{y_j\}_j$ by $\{z_i=i\}_i, \{w_j=j\}_j$ to simplify the approximation process. By doing so, we split the sum into the following parts:
	
\begin{equation}
\label{sum_split}
\begin{split}
	\sum_{s\neq0,n}\frac{p_1^u(s)p_2^u(s)}{p_1^l(s)p_2^l(s)}&=\sum_{|i|<p^{\epsilon}}\sum_{s\in I_i^x}\frac{p_1^u(s)p_2^u(s)}{p_1^l(s)p_2^l(s)}+\sum_{|i|>p^\epsilon}\sum_{j\in J_i^x}\sum_{s\in I_j^y}\frac{(-1)^{i+j}4p^2}{\pi^2}\frac{(\frac{|M_1|s}{p}-i)(\frac{|M_2|s}{p}-j)}{p_1^l(s)p_2^l(s)}\\
	&+\sum_{|i|>p^\epsilon}\sum_{j\in J_i^x: j \text{ odd}}\sum_{s\in I_j^y}\frac{(-1)^{i}4p^2}{\pi^2}\frac{\frac{|M_1|s}{p}-i}{p_1^l(s)p_2^l(s)}
	+\sum_{|i|>p^\epsilon: i \text{ even}}\sum_{j\in J_i^x: j \text{ odd}}\sum_{s\in I_j^y}\frac{4p^2}{\pi^2 s(s+n)}\\
	&=:E_1+S+E_2+E_3.
\end{split}
\end{equation}

	We shall estimate on each of the four terms to show that \eqref{sum_split} is of order $O(p^{3/2-\alpha})$.

\begin{proposition}
\label{est_summary}
	We have the following estimates:
	
\begin{itemize}
\item [(a)]
\[
	E_1=O(p^{3/2-\delta+\epsilon}).
\]
\item[(b)]
\[
	E_2+E_3=O(p^{3/2-\delta}\log(p)).
\]
\item[(c)]
\[
	S=O(p^{3/2-\sigma-\epsilon})+O(p^{3/2-\delta}\log(p)).
\]
\end{itemize}
\end{proposition}

	With the estimates in Proposition \ref{est_summary}, we can prove Theorem \ref{auto_corr_main}:
\begin{proof} of Theorem \ref{auto_corr_main}:

	From Proposition \ref{est_summary}, we see that 
\[
	\sum_{s\neq0,n}\frac{p_1^u(s)p_2^u(s)}{p_1^l(s)p_2^l(s)}=O(p^{3/2-\sigma-\epsilon})+O(p^{3/2-\delta}\log p)+O(p^{3/2-\delta+\epsilon})=O(p^{3/2-\alpha_\epsilon})
\]
	where $\alpha_\epsilon=\min\{\epsilon+\sigma,\delta-\epsilon\}$. Since the choice of $\epsilon$ is arbitrary, we can optimize $\alpha$ to be $\sigma+(\delta-\sigma)/2$, which is what we claimed.

\end{proof}

	We first consider the case when $s$ is positive. The case when $s$ is negative is similar, and the proof for positive indices can be modified verbatim. We consider the term $S$ in \eqref{sum_split} to be the main term, while the rest are considered as correction terms. We shall first compute all three correction terms before dealing with the main term.
	
\section{Estimates of Correction Terms}

	First, we shall prove Proposition \ref{est_summary} (a).

\begin{proof} of Proposition  \ref{est_summary} (a):

	Assuming that $|x_i|\leq p^\epsilon$ and $|M_1|\sim p^{1/2-\sigma}$, $\sigma\in(0,1/2)$, we have $|s|\leq\frac{px_i}{|M_1|}\sim p^{1/2+\epsilon+\sigma}$. Note that $n\sim p^{1/2+\delta}$ where $\delta>\epsilon+\sigma$, $\delta\in(0,1/2)$. Thus,

\begin{equation}
\begin{split}
	\sum_{|s|\leq p^{1/2+\epsilon+\sigma}}|\frac{\sin(\pi|M_1|s/p)}{\sin(\pi s/p)}||\frac{\sin(\pi|M_2|(s+n)/p)}{\sin(\pi (s+n)/p)}|&\leq |M_1|\sum_{|s|\leq p^{1/2+\epsilon+\sigma}}\frac{p}{\pi(s+n)}\\
			&\leq p|M_1|\log(\frac{n+p^{1/2+\epsilon+\sigma}}{n-p^{1/2+\epsilon+\sigma}})\\
			&=p|M_1|\log(1+\frac{1}{np^{-1/2-\epsilon-\sigma}-1})\\
			&\sim p|M_1|p^{-\delta+\epsilon+\sigma}\sim p^{3/2-\delta+\epsilon}.
\end{split}
\end{equation}

	Around the singular point $s=-n$, we make sure to take out an even number of $y_j$-intervals so the cancellations still occur in the remaining $x_i$-interval. Thus, the summation range is $|s+n|\leq \frac{kp}{|M_2|}$ for some $k\in\NN$. Then,
	
\begin{equation}
\begin{split}
	\sum_{|s+n|\leq \frac{p}{|M_2|}}|\frac{\sin(\pi|M_1|s/p)}{\sin(\pi s/p)}||\frac{\sin(\pi|M_2|(s+n)/p)}{\sin(\pi (s+n)/p)}|&\leq |M_2|\sum_{|s+n|\leq p^{1/2+\epsilon+\sigma}}\frac{p}{\pi |s|}\\
			&\leq p|M_2|\log(\frac{n+p/|M_2|}{n-p/|M_2|})\\
			&=p|M_2|\log(1+\frac{1}{n|M_2|/p-1})\\
			&\sim p^{2-1/2-\delta+\epsilon}=p^{3/2-\delta+\epsilon}.
\end{split}
\end{equation}

\end{proof}

	To prove Proposition \ref{est_summary} (b), we need the following lemma:
\begin{lemma}
\label{approx_lemma}
	Let $f,g:\ZZ\to\RR$ be $f(s)=\frac{1}{s}$ and $f(s)=\frac{1}{s(s+t)}$ for some $t\in\RR$. If $1<a<a+1<b$ is such that $\frac{b}{a}=1+r$ for some $r\in(0,1)$, then
\[
	\left\{\begin{array}{lcl}
	\sum_{a\leq s\leq b}f(s)&=&r+O(r^2)+O(b^{-1})\\
	\sum_{a\leq s\leq b}g(s)&=&\frac{r(r+t/a)}{t(1+r+t/a)}+O(\frac{r^2}{t}).
	\end{array}\right.
\] 
\end{lemma} 

\begin{proof}
	Since both $f$ and $g$ are monotone in $(a,b)$, we may approximate the summation of both $f$ and $g$ with their respective integrals. Moreover,
\[
	|\sum_{a\leq s\leq b}\frac{1}{s}-\int_a^b\frac{1}{x}dx|\leq\int_{b-1}^{b+1}\frac{1}{x}dx=\log(1+\frac{2}{b-1})=\frac{2}{b-1}+O(b^{-2}).
\]
	Thus,
\[
	\sum_{a\leq s\leq b}\frac{1}{s}=\log(1+r)+O(b^{-1})=r+O(b^{-1}).
\]

	Note that $\frac{1}{b}<\frac{1}{a}<\frac{b-a}{a}=r$, so $\sum_{a\leq s\leq b}f(s)=O(r)$. For $g$, we have
\[
\begin{split}
	|\sum_{a\leq s\leq b}\frac{1}{s(s+t)}-\int_a^b\frac{1}{s(s+t)}ds|&\leq\int_{b-1}^{b+1}\frac{1}{s(s+t)}ds\\
			&=\frac{1}{t}\log\left(\frac{(b+1)(b+t-1)}{(b-1)(b+t+1)}\right)\\
			&=\frac{1}{t}\log\left((1+\frac{2}{b-1})(1-\frac{2}{b+t+1})\right)\\
			&=\frac{1}{t}(\frac{2}{b-1}-\frac{2}{b+t+1})+O(\frac{1}{tb^2}).
\end{split}
\]

	Thus,
\[
\begin{split}
	\sum_{a\leq s\leq b}\frac{1}{s(s+t)}&=\frac{1}{t}\log\left((1+r)(1-\frac{(b-a)}{b+t})\right)\\
			&=\frac{r}{t}(1-\frac{1}{(1+r)+t/a})+O(\frac{r^2}{t})=\frac{r(r+t/a)}{t(1+r+t/a)}+O(\frac{r^2}{t}).
\end{split}
\]
\end{proof}

\begin{proof} of Proposition \ref{est_summary}:

	Fixing $i$, consider
\begin{equation}
\label{first_correct_term}
\begin{split}
E_2&=p^2\sum_{(2k+1)\in J_i^x}\sum_{s\in I_j^y}\frac{\frac{|M_1|s}{p}-i}{s(s+n)}\\
	&\leq\sum_{s\in I_i^x}\frac{\frac{|M_1|s}{p}-i}{s(s+n)}\\
	&=\sum_{s\in I_i^x}p|M_1|\frac{1}{s+n}-\frac{p^2i}{s(s+n)}\\
	&=\sum_{s\in I_i^x}p|M_1|\frac{1}{s+n}-\frac{p^2i}{n}(\frac{1}{s}-\frac{1}{s+n})\\
	&=p|M_1|\frac{\frac{p}{|M_1|}}{\frac{pi}{|M_1|}+n}+O(p|M_1|r^2)+O(|M_1|^2 i^{-1})\\
	&-\frac{p^2i}{n}\left(\frac{1}{i}-\frac{\frac{p}{|M_1|}}{\frac{pi}{|M_1|}+n}\right)+O(p^2n^{-1}x_i^{-1})\\
	&=p|M_1|\frac{1}{i+\frac{|M_1|n}{p}}-\frac{p^2i}{n}\frac{\frac{|M_1|n}{p}}{i(i+\frac{|M_1|n}{p})}+O(p^2n^{-1}x_i^{-1})\\
	&=O(p^2 n^{-1}i^{-1}).
\end{split}
\end{equation}
	
	As for $E_3$, by Lemma \ref{approx_lemma},
\begin{equation}
\label{second_correct_term}
\begin{split}
	E_3&=p^2\sum_{i\in2\ZZ+1: i<|M_1|/2}\sum_{(2k+1)\in J_i^x}\sum_{s\in I_j^y}\frac{1}{s(s+n)}\\
	&\leq \sum_{i<|M_1|/2}p^2\sum_{s\in I_i^x}\frac{1}{s(s+n)}\\
	&\leq \sum_{i<|M_1|/2}\frac{p^2}{ni}+O(p^2n^{-1}i^{-2})=O(p^2n^{-1}\log(|M_1|)).
\end{split}
\end{equation}
	Combining \eqref{first_correct_term} and \eqref{second_correct_term}, we see that the total contribution is $p^{2-(1/2+\delta)}\log(|M_1|)= O(p^{3/2-\delta}\log(p))$.
\end{proof}

\section{Estimates of the Main Term}

	To estimate $S$ in \eqref{sum_split}, we start by computing the expression of the sum in one $y_j$-interval.

\subsection{Estimates within $y_j$-Intervals}

\begin{lemma}
\label{y_j_contribute}
	Given $j>0$, the sum of the main term within the $y_j$-interval satisfies
\[
	E_y(j):=\sum_{s\in I_j^y}\frac{4p^2}{\pi^2}\frac{(\frac{|M_1|s}{p}-i)(\frac{|M_2|(s+n)}{p}-j)}{s(s+n)}-\left(\frac{-2p^3i}{\pi^2|M_2|\tilde{y}_j^2}+\frac{2np^2i}{\tilde{y}_j^2 j}+\frac{2p|M_1|}{\pi^2j}\right)=O(p^{3/2-\alpha}),
\]
	where $\tilde{y}_j=\frac{pj}{|M_2|}-n$, and  $\alpha$ is as defined in Theorem \ref{auto_corr_main}.
\end{lemma}

\begin{proof}

	For $s\in I_j^y$ all $p_1^u, p_2^u, p_1^l, p_2^l$ are linear and none changes sign. Thus, 
\[
\sum_{s\in I_j^y}\frac{p_1^u(s)p_2^u(s)}{p_1^l(s)p_2^l(s)}=\pm\sum_{s\in I_j^y}\frac{4p^2}{\pi^2}\frac{(\frac{|M_1|s}{p}-i)(\frac{|M_2|(s+n)}{p}-j)}{s(s+n)}.
\]	
	Thus, we would like to compute
\[
\begin{split}
	&\sum_{s\in I_j^y}\frac{4p^2}{\pi^2}\frac{(\frac{|M_1|s}{p}-i)(\frac{|M_2|(s+n)}{p}-j)}{s(s+n)}\\
	&=\sum_{\frac{pj}{|M_2|}-n\leq s\leq \frac{p{j+1}}{|M_2|}-n}\bigg\{\frac{4}{\pi^2}|M_1||M_2|-\frac{4p|M_1|j}{\pi^2}\frac{1}{s+n}-\frac{4p|M_2|i}{\pi^2}\frac{1}{s}+\frac{4p^2}{\pi^2}\frac{i j}{s(s+n)}    \bigg\}\\
	&=\frac{4|M_1||M_2|}{\pi^2}(\tilde{y}_{j+1}-\tilde{y}_j+f_1(j))-\frac{4p|M_1|j}{\pi^2}(\log(\frac{{j+1}}{j})+f_2(j))-\frac{4p|M_2|i}{\pi^2}(\log(\frac{\tilde{y}_{j+1}}{\tilde{y}_j})+f_3(j))\\
	&+\frac{4p^2 ij}{n\pi^2}\bigg(\log(\frac{\tilde{y}_{j+1}}{\tilde{y}_j})-\log(\frac{\tilde{y}_{j+1}+n}{\tilde{y}_j+n})+f_4(j)\bigg),
\end{split}
\]
	where we recall that $\tilde{y}_j=\frac{py_j}{|M_2|}-n$.

	Note that $\tilde{y}_{j+1}-\tilde{y}_j=\frac{p}{|M_2|}$, and also $\log(\frac{{j+1}}{j})=\log(1+\frac{1}{j})=\frac{1}{j}-\frac{1}{2j^2}+O(j^{-3})$. Thus,
\[
	\frac{4|M_1||M_2|}{\pi^2}(\tilde{y}_{j+1}-\tilde{y}_j)-\frac{4p|M_1|j}{\pi^2}\log(\frac{{j+1}}{j})=\frac{2p|M_1|}{\pi^2j}+O(|M_1|py_j^{-2}).
\]

	Now, 
\[
\begin{split}
\log(\frac{\tilde{y}_{j+1}}{\tilde{y}_j})=\log(1+\frac{\frac{p}{|M_2|}}{\tilde{y}_j})&=\frac{p}{|M_2|}\frac{1}{\frac{pj}{|M_2|}-n}-\frac{p^2}{2|M_2|^2}\frac{1}{\tilde{y}_j^2}+O(p^3|M_2|^{-3}\tilde{y}_j^{-3})\\
			&=\frac{p}{|M_2|}\frac{1}{\frac{pj}{|M_2|}-n}-\frac{p^2}{2|M_2|^2}\frac{1}{\tilde{y}_j^2}+\frac{p^3}{3|M_2|^3}\frac{1}{\tilde{y}_j^3}+O(p^4|M_2|^{-4}\tilde{y}_j^{-4}),
\end{split}
\] 
	and 
\begin{equation}
\label{tilde_log}
\begin{split}
\log(\frac{\tilde{y}_{j+1}+n}{\tilde{y}_j+n})=\log(\frac{{j+1}}{j})&=\frac{1}{j}-\frac{1}{2j^2}+O(j^{-3})\\
			&=\frac{1}{j}-\frac{1}{2j^2}+\frac{1}{3j^3}+O(j^{-4}).
\end{split}
\end{equation}
	Thus, we see that, using \eqref{tilde_log}, 
\[
\begin{split}
	&-\frac{4p|M_2|i}{\pi^2}\frac{p}{|M_2|}\frac{1}{\frac{pj}{|M_2|}-n}+\frac{4p^2ij}{n\pi^2}\frac{p}{|M_2|}\frac{1}{\frac{pj}{|M_2|}-n}\\
	&=\frac{-4p^2j}{\pi^2}\frac{1}{\frac{pj}{|M_2|}-n}+\frac{4p^2i}{n\pi^2}+\frac{4p^2i}{\pi}\frac{1}{\frac{pj}{|M_2|}-n}\\
	&=\frac{4p^2i}{n\pi^2}.
\end{split}
\]
	Combining all of the above, we get that
\begin{equation}
\label{first_extract}
\begin{split}
	&\frac{4|M_1||M_2|}{\pi^2}(\tilde{y}_{j+1}-\tilde{y}_j+f_1(j))-\frac{4p|M_1|j}{\pi^2}(\log(\frac{{j+1}}{j})+f_2(j))-\frac{4p|M_2|i}{\pi^2}(\log(\frac{\tilde{y}_{j+1}}{\tilde{y}_j})+f_3(j))\\
	&+\frac{4p^2 ij}{n\pi^2}\bigg(\log(\frac{\tilde{y}_{j+1}}{\tilde{y}_j})-\log(\frac{\tilde{y}_{j+1}+n}{\tilde{y}_j+n})+f_4(j)\bigg)\\
	&=\frac{4|M_1||M_2|}{\pi^2}f_1(j)+\frac{2p|M_1|}{\pi^2 j}+O(|M_1|pj^{-2})-\frac{4p|M_1|j}{\pi^2}f_2(j)+\frac{2p^3i}{\pi^2|M_2|\tilde{y}_j^2}-\frac{4p|M_2|i}{\pi^2}f_3(j)\\
	&+O(|M_3|^{-2}p^4\tilde{y}_j^{-3})+\frac{4p^2i}{n\pi^2}-\frac{2p^4ij}{n|M_2|^2\pi^2\tilde{y}_j^2}-\frac{4p^2i}{n\pi^2}+\frac{2p^2i}{n\pi^2j}+\frac{4p^2ij}{n\pi^2}f_4(j)+\frac{4p^{2}ij}{3n\pi^2}f_4'(j)\\
	&+O(n^{-1}p^6|M_2|^{-4}ij\tilde{y}_j^{-4})+O(n^{-1}p^2ij^{-3})\\
	&=-\frac{2p^4ij}{n|M_2|^2\pi^2\tilde{y}_j^2}+\frac{2p^2x_i}{n\pi^2y_j}+\frac{2p|M_1|}{\pi^2 y_j}+\frac{2p^3x_i}{\pi^2|M_2|\tilde{y}_j^2}\\
	&+\frac{4|M_1||M_2|}{\pi^2}f_1(j)-\frac{4p|M_1|j}{\pi^2}f_2(j)-\frac{4p|M_2|i}{\pi^2}f_3(j)+\frac{4p^2ij}{n\pi^2}f_4(j)+\frac{4p^{2}ij}{3n\pi^2}f_4'(j)\\
	&+O(|M_1|pj^{-2})+O(|M_2|^{-2}p^4\tilde{y}_j^{-3})+O(n^{-1}p^6|M_2|^{-4}ij\tilde{y}_j^{-4})+O(n^{-1}p^2ij^{-3}),
\end{split}
\end{equation}
	where
	
\[
\begin{split}
	\frac{4p^{2}ij}{3n\pi^2}f_4'(j)&=\frac{4p^{2}ij}{3n\pi^2}\bigg(\frac{p^3}{|M_2|^3}\frac{1}{\tilde{y}_j^3}-\frac{1}{j^3}\bigg)\\
		&=\frac{4p^2ij}{3n\pi^2}\bigg(\frac{3\frac{p^2}{|M_2|^2}j^2n-3\frac{p}{|M_2|}jn^2+n^3}{\tilde{y}_j^3 j^3}\big)\\
		&=O(p^4|M_2|^{-2}ij^{-1}\tilde{y}_j^{-3})+O(p^3|M_2|^{-1}nij^{-2}\tilde{y}_j^{-3})+O(p^2n^2ij^{-3}\tilde{y}_j^{-3}).
\end{split}
\]

	In \eqref{first_extract}, we have four explicit terms remaining, namely
\begin{equation}
\label{explicit_terms}
	\frac{-2p^4ij}{n\pi^2|M_2|^2\tilde{y}_j^2}+\frac{2p^2i}{n\pi^2j}+\frac{2p|M_1|}{\pi^2j}+\frac{2p^3i}{\pi^2|M_2|\tilde{y}_j^2}.
\end{equation}
		
	Further simplifying the expressions, we have
	
\[
\begin{split}
	\frac{-2p^4ij}{n\pi^2|M_2|^2\tilde{y}_j^2}+\frac{2p^2i}{n\pi^2j}+\frac{2p^3i}{\pi^2|M_2|\tilde{y}_j^2}
	&=\frac{2p^2}{n\pi^2}\bigg[\frac{\tilde{y}_j^2-\frac{p^2}{|M_2|^2}j^2}{\tilde{y}_j^2 j}\bigg]+\frac{2p^3i}{\pi^2|M_2|\tilde{y}_j^2}\\
	&=\frac{2p^2}{n\pi^2}\frac{-\frac{2p}{|M_2|}nj+n^2}{\tilde{y}_j^2 j}+\frac{2p^3i}{\pi^2|M_2|\tilde{y}_j^2}\\
	&=\frac{-4p^3i}{\pi^2|M_2|\tilde{y}_j^2}+\frac{2np^2i}{\tilde{y}_j^2 j}+\frac{2p^3i}{\pi^2|M_2|\tilde{y}_j^2}\\
	&=\frac{-2p^3i}{\pi^2|M_2|\tilde{y}_j^2}+\frac{2np^2i}{\tilde{y}_j^2 j}.
\end{split}
\]

	To this point, we have computed all the main terms, and we have
\begin{equation}
\label{E_expand}
\begin{split}
	E(j)=&\frac{4|M_1||M_2|}{\pi^2}f_1(j)-\frac{4p|M_1|j}{\pi^2}f_2(j)-\frac{4p|M_2|i}{\pi^2}f_3(j)+\frac{4p^2ij}{n\pi^2}f_4(j)+\frac{4p^{2}ij}{3n\pi^2}f_4'(j)\\
	&+O(|M_1|pj^{-2})+O(|M_2|^{-2}p^4\tilde{y}_j^{-3})+O(n^{-1}p^6|M_2|^{-4}ij\tilde{y}_j^{-4})+O(n^{-1}p^2ij^{-3}).
\end{split}
\end{equation}
	To estimate the effect of $f_1,f_2,f_3,f_4,f_4'$, we refer to the following proposition which shall be proved in Section \ref{error_term}.
\begin{proposition}
\label{error_f_prop}
	The following estimates hold:
\begin{itemize}
\item
\[
\sum_{|i|>p^\epsilon}\sum_{y_j\in J^x_i}\frac{4p^{2}ij}{3n\pi^2}|f_4'(j)|=\sum_{|i|>p^\epsilon}O(p|M_1|i^{-3})=O(p^{3/2-\sigma-2\epsilon}).
\]		
\item
\[
	\sum_{i=-|M_1|/2}^{|M_1|/2}\sum_{j\in J_i^x}\frac{4p|M_2|i}{\pi^2}|f_3(j)|+\frac{4p^2ij}{n\pi^2}|f_4(j)|=O(|M_1||M_2|\log p).
\]
\item
\[
\sum_{y_j=1}^{|M_2|}\frac{4p|M_1|j}{\pi^2}f_2(j)=O(|M_1||M_2|\log |M_2|).
\]

\item
	$\sum_{x\geq p^\epsilon}\sum_{y\in J_i^x}|M_1||M_2|f_1(j)=O(|M_2|^2)=O(p)$ if $|M_1|$ is even.
\end{itemize}
\end{proposition}

	Proposition \ref{error_f_prop} shows that the first five terms in \eqref{E_expand} sums up to be of the order $O(p^{3/2-\sigma-2\epsilon})$. Thus, it remains to show that the final four terms in \eqref{E_expand} can be well controlled.
	
	Note that $p|M_2|^-1\tilde{y}_j^{-1}=\frac{1}{j-\frac{n|M_2|}{p}}=:\frac{1}{j-t}$. Thus,
\begin{itemize}
\item
\[
	\sum_{|i|>p^\epsilon}\sum_{j\in J_i^x}p|M_1|j^{-2}=p|M_1|\sum_{|j|>\frac{|M_2|}{|M_1|}p^{\epsilon}}j^{-2}= O(p^{3/2-\sigma-\epsilon}).
\]
\item
\[
	\sum_{|i|>p^\epsilon}\sum_{j\in J_i^x} p^4|M_2|^{-2}\tilde{y}_j^{-3}=p|M_2|^{-1}\sum_{|j|>\frac{|M_2|}{|M_1|}p^\epsilon}\frac{1}{(j-t)^3}=O(p^{1/2}).
\]
\item
\[
\begin{split}
	\sum_{|i|>p^\epsilon}\sum_{j\in J_i^x} n^{-1}p^6|M_2|^{-4}ij\tilde{y}_j^{-4}&=\frac{p^2}{n}\sum_{|i|>p^\epsilon}\sum_{j\in J_i^x}i\left(\frac{1}{(j-t)^4}+\frac{t}{(j-t)^3}\right)\\
			&=\frac{p^2}{n}\sum_{|i|>p^\epsilon}\sum_{k=\frac{|M_2|i}{|M_1|}}^{\frac{|M_2|(i+1)}{|M_1|}}i\left(\frac{1}{k^4}+\frac{t}{k^3}\right)\\
			&\leq\frac{p^2}{n}\sum_{|i|>p^\epsilon}\frac{3|M_1|^3}{|M_2|^3}\frac{i((i+1)^3-i^3)}{i^3(i+1)^3}+\frac{2|M_1|^2}{|M_2|^2}\frac{it((i+1)^2-i^2)}{i^2(i+1)^2}\\
			&\leq\frac{20p^2}{n}\frac{|M_1|^3}{|M_2|^3}p^{-2\epsilon}+\frac{|M_1|^2}{|M_2|^2}\frac{n|M_2|}{p}p^{-\epsilon}\\
			&=O(\frac{20p^2}{n}(p^{-2\epsilon}+p^{\delta-\sigma-\epsilon}))=O(p^{3/2-\sigma-\epsilon}).
\end{split}
\]
\item
\[
	\sum_{|i|>p^\epsilon}\sum_{j\in J_i^x} n^{-1}p^2 ij^{-3}\leq\frac{p^2}{n}\sum_{|i|>p^\epsilon}\frac{|M_1|^2i}{|M_2|^2}\left(\frac{1}{i^2}-\frac{1}{(i+1)^2}\right)=\frac{p^2}{n}\sum_{|i|>p^\epsilon}\frac{|M_1|^2}{|M_2|^2}\frac{2i+1}{i(i+1)^2}= O(p^{3/2-\delta-\epsilon}).
\]
\end{itemize}
	Combining all the terms above, we see that $\sum_{|i|>p^\epsilon}\sum_{j\in J_i^x}E(j)=O(p^{3/2-\sigma-\epsilon})$. Choosing $\epsilon=(\delta-\sigma)/2$, we see that it is indeed of the order $p^{3/2-\alpha}$.
\end{proof}

%=====================================

\subsection{Estimates within $x_i$-Intervals}

	Within a given $I_i^x$, $p_1^u, p_1^l, p_2^l$ do not change signs, but $p_2^u$ does between $I_j^y$ and $I_{j+1}^y$. Thus, the main terms in Lemma \ref{y_j_contribute} flip signs across different $y_j$-intervals.
	
	Note that between consecutive $y_j$-intervals, either $y_{j+1}=y_j$ or $y_{j+1}=y_j+2$ by construction. Moreover, $y_0=x_0=0$. In this section, we replace $\{y_j\}_j$ by $\{z_j\}_j$ where $z_j=j$. Then, we have
\[
	y_j-z_j=\left\{\begin{array}{lcl}
	1&\text{if}& j\in2\ZZ+1\\
	0&\text{if}& j\in2\ZZ.
	\end{array}\right.
\]
	In particular, we may split the sum into
\[
	\sum_{j\in J_i^x}\sum_{s\in I_j^y}\frac{p^u_1(s)p^u_2(s)}{p^l_1(s)p^l_2(s)}=\sum_{j\in J_i^x}(-1)^j F(j)+\frac{4p^2}{\pi^2}\sum_{(2k+1)\in J_i^x}\sum_{s\in I_j^y}\frac{\frac{|M_1|s}{p}-x_j}{s(s+n)},
\]
	where 
\[
F(j)=\sum_{s\in I_j^y}\frac{(\frac{|M_1|x}{p}-x_j)(\frac{|M_2|(s+n)}{p}-j)}{s(s+n)}.
\]

	We present the following approximation lemma:

\begin{lemma}
\label{approx_lemma}
	Let $f,g:\ZZ\to\RR$ be $f(s)=\frac{1}{s}$ and $f(s)=\frac{1}{s(s+t)}$ for some $t\in\RR$. If $1<a<a+1<b$ is such that $\frac{b}{a}=1+r$ for some $r\in(0,1)$, then
\[
	\left\{\begin{array}{lcl}
	\sum_{a\leq s\leq b}f(s)&=&r+O(r^2)+O(b^{-1})\\
	\sum_{a\leq s\leq b}g(s)&=&\frac{r(r+t/a)}{t(1+r+t/a)}+O(\frac{r^2}{t}).
	\end{array}\right.
\] 
\end{lemma} 

\begin{proof}
	Since both $f$ and $g$ are monotone in $(a,b)$, we may approximate the summation of both $f$ and $g$ with their respective integrals. Moreover,
\[
	|\sum_{a\leq s\leq b}\frac{1}{s}-\int_a^b\frac{1}{x}dx|\leq\int_{b-1}^{b+1}\frac{1}{x}dx=\log(1+\frac{2}{b-1})=\frac{2}{b-1}+O(b^{-2}).
\]
	Thus,
\[
	\sum_{a\leq s\leq b}\frac{1}{s}=\log(1+r)+O(b^{-1})=r+O(b^{-1}).
\]

	Note that $\frac{1}{b}<\frac{1}{a}<\frac{b-a}{a}=r$, so $\sum_{a\leq s\leq b}f(s)=O(r)$. For $g$, we have
\[
\begin{split}
	|\sum_{a\leq s\leq b}\frac{1}{s(s+t)}-\int_a^b\frac{1}{s(s+t)}ds|&\leq\int_{b-1}^{b+1}\frac{1}{s(s+t)}ds\\
			&=\frac{1}{t}\log\left(\frac{(b+1)(b+t-1)}{(b-1)(b+t+1)}\right)\\
			&=\frac{1}{t}\log\left((1+\frac{2}{b-1})(1-\frac{2}{b+t+1})\right)\\
			&=\frac{1}{t}(\frac{2}{b-1}-\frac{2}{b+t+1})+O(\frac{1}{tb^2}).
\end{split}
\]

	Thus,
\[
\begin{split}
	\sum_{a\leq s\leq b}\frac{1}{s(s+t)}&=\frac{1}{t}\log\left((1+r)(1-\frac{(b-a)}{b+t})\right)\\
			&=\frac{r}{t}(1-\frac{1}{(1+r)+t/a})+O(\frac{r^2}{t})=\frac{r(r+t/a)}{t(1+r+t/a)}+O(\frac{r^2}{t}).
\end{split}
\]
\end{proof}
	
	With this lemma, we have the following corollary:

\begin{corollary}
	The contribution of the correction term satisfies
\[
\frac{4p^2}{\pi^2}\sum_{i>p^\epsilon}\sum_{(2k+1)\in J_i^x}\sum_{s\in I_j^y}\frac{(-1)^i(\frac{|M_1|s}{p}-x_j)}{s(s+n)}=O(p^{3/2-\delta}\log(p)).
\]
\end{corollary}
	
\begin{proof}
	Fixing $i$, consider the inner sum
\begin{equation}
\label{first_correct_term}
\begin{split}
p^2\sum_{(2k+1)\in J_i^x}\sum_{s\in I_j^y}\frac{\frac{|M_1|s}{p}-i}{s(s+n)}&\leq\sum_{s\in I_i^x}\frac{\frac{|M_1|s}{p}-i}{s(s+n)}\\
	&=\sum_{s\in I_i^x}p|M_1|\frac{1}{s+n}-\frac{p^2i}{s(s+n)}\\
	&=\sum_{s\in I_i^x}p|M_1|\frac{1}{s+n}-\frac{p^2i}{n}(\frac{1}{s}-\frac{1}{s+n})\\
	&=p|M_1|\frac{\frac{p}{|M_1|}}{\frac{pi}{|M_1|}+n}+O(p|M_1|r^2)+O(|M_1|^2 i^{-1})\\
	&-\frac{p^2i}{n}\left(\frac{1}{i}-\frac{\frac{p}{|M_1|}}{\frac{pi}{|M_1|}+n}\right)+O(p^2n^{-1}x_i^{-1})\\
	&=p|M_1|\frac{1}{i+\frac{|M_1|n}{p}}-\frac{p^2i}{n}\frac{\frac{|M_1|n}{p}}{i(i+\frac{|M_1|n}{p})}+O(p^2n^{-1}x_i^{-1})\\
	&=O(p^2 n^{-1}i^{-1}).
\end{split}
\end{equation}
	Again, we approximated $\{x_i\}_i$ by $\{w_i=i\}_i$. The contribution of the difference is, by Lemma \ref{approx_lemma},
\begin{equation}
\label{second_correct_term}
\begin{split}
	p^2\sum_{i\in2\ZZ+1: i<|M_1|/2}\sum_{(2k+1)\in J_i^x}\sum_{s\in I_j^y}\frac{1}{s(s+n)}&\leq \sum_{i<|M_1|/2}p^2\sum_{s\in I_i^x}\frac{1}{s(s+n)}\\
	&\leq \sum_{i<|M_1|/2}\frac{p^2}{ni}+O(p^2n^{-1}i^{-2})=O(p^2n^{-1}\log(|M_1|)).
\end{split}
\end{equation}
	Combining \eqref{first_correct_term} and \eqref{second_correct_term}, we see that the total contribution is $p^{2-(1/2+\delta)}\log(|M_1|)\sim p^{3/2-\delta}\log(p)$.
\end{proof}
	
	$|M_2|n/p$ will not be an integer unless $n=0$. Suppose for now that $|M_1|\vert|M_2|$. Then we see that there will be $|M_2|/|M_1|-1$ complete y-intervals within. Also, the left and right incomplete y-intervals will combine to have the same length of a complete y-interval. 
	
	Define $g_1(y_j)=\frac{-2p^3x_i}{\pi^2|M_2|\tilde{y}_j^2}, g_2(y_j)=\frac{2np^2x_i}{\tilde{y}_j^2 y_j}, g_3(y_j)=\frac{2p|M_1|}{\pi^2y_j}$. All three terms are decreasing with respect to $y_j$. Thus,
\[
	\sum_{y\in J_i^x}\sum_{s\in I^y_j}\frac{p^u_1(s)p^u_2(s)}{p^l_1(s)p^l_2(s)}\leq\sum_{l=1}^3|\sum_{y_j\in J_i^x}(-1)^{y_j}g_l(y_j)|+\sum_{y_j\in J_i^x}|E(y_j)|
\]
	consists of three alternating series.
	
	Recal that $I_j^y\subset I_i^x\iff y_j\in[\frac{x_i|M_2|}{|M_1|}+\frac{|M_2|n}{p},\frac{x_{i+1}|M_2|}{|M_1|}+\frac{|M_2|n}{p}-1]$, and $\big|J_x^i\big|=|M_2|/|M_1|$. Thus, the case when $|M_2|/|M_1|$ is an even number will be superior to the one with odd numbers.
		
	With the three terms carrying over, we need the following lemma:
	
\begin{lemma}
	Within an $x_i$-interval, the contribution is

\[
	\sum_{s\in I^x_i}\frac{p^u_1(s)p^u_2(s)}{p^l_1(s)p^l_2(s)}=O(p|M_1|i^{-2})+O(p^{3/2-\delta}i^{-1})+\sum_{y\in J_i^x}E_y(j).
\]

%\[
%	\sum_{y_j=\frac{x_i|M_2|}{|M_1|}+\frac{|M_2|n}{p}}^{\frac{x_{i+1}|M_2|}{|M_1|}+\frac{|M_2|n}{p}}\frac{-2p^3x_i}{\pi^2|M_2|\tilde{y}_j^2}+\frac{2np^2x_i}{\tilde{y}_j^2 y_j}+\frac{2p|M_1|}{\pi^2y_j}=O(p|M_1|x_i^{-2})+O(p^{3/2-\delta}x_i^{-1}).
%\]	

\end{lemma}
	
\begin{proof}
	
	First, note that
	
\[
\begin{split}
	&\Bigg|\sum_{j=\frac{i|M_2|}{|M_1|}+\frac{|M_2|n}{p}+1}^{\frac{{i+1}|M_2|}{|M_1|}+\frac{|M_2|n}{p}}(-1)^{j}\bigg(\frac{-2p^3i}{\pi^2|M_2|\tilde{y}_j^2}+\frac{2np^2i}{\tilde{y}_j^2 j}+\frac{2p|M_1|}{\pi^2j}\bigg)\Bigg|\\
	&\leq|\sum_{s\in I^x_i}\frac{p^u_1(s)p^u_2(s)}{p^l_1(s)p^l_2(s)}|\\
	&\leq\Bigg|\sum_{j=\frac{i|M_2|}{|M_1|}+\frac{|M_2|n}{p}}^{\frac{{i+1}|M_2|}{|M_1|}+\frac{|M_2|n}{p}-1}(-1)^{j}\bigg(\frac{-2p^3i}{\pi^2|M_2|\tilde{y}_j^2}+\frac{2np^2i}{\tilde{y}_j^2 j}+\frac{2p|M_1|}{\pi^2j}\bigg)\Bigg|.
\end{split}
\]
	
	Since $|M_2|/|M_1|\in2\NN$, we can see that, for $\frac{-2p^3 i}{\pi^2|M_2|\tilde{y}_j^2}$,

\[
\begin{split}
	\sum_{j=\frac{i|M_2|}{|M_1|}+\frac{|M_2|n}{p}}^{\frac{{i+1}|M_2|}{|M_1|}+\frac{|M_2|n}{p}-1}\frac{(-1)^{j}}{\tilde{y}_j^2}&=\frac{|M_2|^2}{p^2}\sum_{j}\frac{(-1)^{j}}{(j-\frac{|M_2|n}{p})^2}\\
			&=\frac{|M_2|^2}{p^2}\sum_{z_j=\frac{i|M_2|}{|M_1|}}^{\frac{{i+1}|M_2|}{|M_1|}}\frac{(-1)^{j}}{z_j^2}\\
			&\leq\frac{|M_2|^2}{2p^2}\bigg[\bigg(\frac{|M_1|}{i|M_2|}-\frac{|M_1|}{(i+1)|M_2|}\bigg)-\bigg(\frac{1}{\frac{i|M_2|}{|M_1|}+1}-\frac{1}{\frac{(i+1)|M_1|}{|M_2|}+1}\bigg)\bigg]\\
			&=\frac{|M_2|^2}{2p^2}\bigg[\frac{1}{\frac{i|M_2|}{|M_1|}(\frac{i|M_2|}{|M_1|}+1)}-\frac{1}{\frac{(i+1)|M_2|}{|M_1|}(\frac{(i+1)|M_2|}{|M_1|}+1)}\bigg]\\
			&=\frac{|M_2|^2}{2p^2}\bigg[\frac{\frac{|M_2|}{|M_1|}(\frac{i|M_2|}{|M_1|}+\frac{i|M_2|}{|M_1|}+1)+\frac{|M_2|^2}{|M_1|^2}}{\frac{i|M_2|}{|M_1|}(\frac{i|M_2|}{|M_1|}+1)\frac{(i+1)|M_2|}{|M_1|}(\frac{(i+1)|M_2|}{|M_1|}+1)}\bigg]\\
			&=O(p^{-2}|M_1|^2 i^{-3}).
\end{split}
\]

	For $\frac{2np^2 i}{\tilde{y}_j^2 j}$,

\[
	\sum_{j\in J_i^x}\frac{(-1)^{j}}{j\tilde{y}_j^2}=\sum_{j\in J_i^x}(-1)^{j}\bigg[\frac{Aj+B}{\tilde{y}_j^2}+\frac{C}{j}\bigg]
\]
	where $A, B, C$ satisfy
\[
	C\tilde{y}_j^2+Aj^2+Bj=1\implies A=\frac{-p^2}{|M_2|^2 n^2},\quad B=\frac{2p}{n|M_2|},\quad C=\frac{1}{n^2}.
\]
	
	Thus, we have
\[
\begin{split}
	\sum_{j}\frac{(-1)^{j}}{j\tilde{y}_j^2}&=\sum_{j}(-1)^{j}\bigg[\frac{\frac{-p}{|M_2|n^2}\tilde{y}_j+\frac{p}{n|M_2|}}{\tilde{y}_j^2}+\frac{1}{n^2j}\bigg]\\
		&=\sum_{j}(-1)^{j}\bigg[\frac{-p}{|M_2|n^2}\frac{1}{\tilde{y}_j}+\frac{1}{n^2 j}+\frac{p}{n|M_2|}\frac{1}{\tilde{y}_j^2}\bigg]\\
		&\sim\frac{|M_2|}{p}\frac{-p}{|M_2|n^2}\bigg[\log(\frac{i+1}{i})-\log(\frac{\frac{(i+1)|M_2|}{|M_1|}+1}{\frac{i|M_2|}{|M_1|}+1})\bigg]\\
		&+\frac{1}{n^2}\bigg[\log(\frac{\frac{(i+1)|M_2|}{|M_1|}+\frac{|M_2|n}{p}}{\frac{i|M_2|}{|M_1|}+\frac{|M_2|n}{p}})-\log(\frac{\frac{(i+1)|M_2|}{|M_1|}+\frac{|M_2|n}{p}+1}{\frac{i|M_2|}{|M_1|}+\frac{|M_2|n}{p}+1})\bigg]+O(p^{-1}|M_1|^2|M_2|^{-1}n^{-1}i^{-3})\\
		&\sim\frac{1}{n^2}\bigg[-\frac{1}{i}+\frac{1}{i+\frac{|M_1|}{|M_2|}}+\frac{1}{i+\frac{|M_1|n}{p}}-\frac{1}{i+\frac{|M_1|n}{p}+\frac{|M_1|}{|M_2|}}\bigg]+O(p^{-1}|M_1|^2|M_2|^{-1}n^{-1}i^{-3})\\
		&=\frac{|M_1|}{n^2|M_2|}\bigg[\frac{(2i+\frac{|M_1|}{|M_2|})\frac{|M_1|n}{p}+\frac{|M_1|^2 n^2}{p^2}}{i(i+\frac{|M_1|}{|M_2|})(i+\frac{|M_1|n}{p})(i+\frac{|M_1|n}{p}+\frac{|M_1|}{|M_2|})}\bigg]\\
		&=\frac{|M_1|}{|M_2|n^2}\frac{O(\frac{i|M_1|n}{p})+O(\frac{|M_1|^2n^2}{p^2})}{O(i^4)+O(i^2\frac{|M_1|^2n^2}{p^2})}.
\end{split}
\]

	For $\frac{2p|M_1|}{\pi^2 j}$, by letting $2a=\frac{i|M_2|}{|M_1|}+\frac{|M_2|n}{p}$, $2b=\frac{(i+1)|M_2|}{|M_1|}+\frac{|M_2|n}{p}$, and $t=1/2$, we have
\[
\begin{split}
	\sum_{j\in J_i^x}\frac{(-1)^{j}}{j}&=\sum_{(2k)\in J_i^x}(\frac{1}{2k}-\frac{1}{2k+1})\\
				&=\sum_{(2k)\in J_i^x}\frac{1}{2k(2k+1)}\\
				&=\frac{\frac{1}{i+\frac{|M_1|n}{p}}\left(\frac{1}{i+\frac{|M_1|n}{p}}+\frac{1}{\frac{i|M_2|}{|M_1|}+\frac{|M_2|n}{p}}\right)}{1+\frac{1}{i+\frac{|M_1|n}{p}}+\frac{1}{\frac{i|M_2|}{|M_1|}+\frac{|M_2|n}{p}}}+O(r^2)\\
				&\leq \frac{1}{(i+\frac{|M_1|n}{p})^2}+O(r^2)=O(\frac{1}{i^2}).
\end{split}
\]

	Combining the three terms, we see that
\[
\begin{split}
	&\sum_{j\in J_i^x}(-1)^{j}\bigg(\frac{-2p^3i}{\pi^2|M_2|\tilde{y}_j^2}+\frac{2np^2i}{\tilde{y}_j^2 j}+\frac{2p|M_1|}{\pi^2j}\bigg)\\
	&=O(p|M_1|^2|M_2|^{-1}i^{-2})+\min\{O(i^{-2}p^{3/2-\sigma}), O(p^{3/2-2\delta+\sigma})\}\\
	&+\min\{O(i^{-3}p^{3/2-2\sigma+\delta}), O(p^{3/2-\delta}i^{-1})\}+O(p|M_1|^2|M_2|^{-1}i^{-2})\\
	&=O(p|M_1|i^{-2})+O(p^{3/2-\delta}i^{-1}).
\end{split}
\]

\end{proof}

\subsection{Main Terms Estimates}

	Now, we are prepared to prove our theorem.

\begin{proof} of Theorem \ref{auto_corr_main}:

	The calculations above accounts for most of the intervals, but we need be more careful around the singularities of $p_1^l$ and $p_2^l$, namely $s\sim0$ and $s\sim -n$.
	
	Let's suppose $|x_i|\leq p^\epsilon$ and $|M_1|\sim p^{1/2-\sigma}$, $\sigma\in(0,1/2)$, then $|s|\leq\frac{px_i}{|M_1|}\sim p^{1/2+\epsilon+\sigma}$. Note that $n\sim p^{1/2+\delta}$ where $\delta>\epsilon+\sigma$, $\delta\in(0,1/2)$. Thus,

\begin{equation}
\begin{split}
	\sum_{|s|\leq p^{1/2+\epsilon+\sigma}}|\frac{\sin(\pi|M_1|s/p)}{\sin(\pi s/p)}||\frac{\sin(\pi|M_2|(s+n)/p)}{\sin(\pi (s+n)/p)}|&\leq |M_1|\sum_{|s|\leq p^{1/2+\epsilon+\sigma}}\frac{p}{\pi(s+n)}\\
			&\leq p|M_1|\log(\frac{n+p^{1/2+\epsilon+\sigma}}{n-p^{1/2+\epsilon+\sigma}})\\
			&=p|M_1|\log(1+\frac{1}{np^{-1/2-\epsilon-\sigma}-1})\\
			&\sim p|M_1|p^{-\delta+\epsilon+\sigma}\sim p^{3/2-\delta+\epsilon}.
\end{split}
\end{equation}

	Around the singular point $s=-n$, we make sure to take out an even number of $y_j$-intervals so the cancellations still occur in the remaining $x_i$-interval. Thus, the summation range is $|s+n|\leq \frac{kp}{|M_2|}$ for some $k\in\NN$. Then,
	
\begin{equation}
\begin{split}
	\sum_{|s+n|\leq \frac{p}{|M_2|}}|\frac{\sin(\pi|M_1|s/p)}{\sin(\pi s/p)}||\frac{\sin(\pi|M_2|(s+n)/p)}{\sin(\pi (s+n)/p)}|&\leq |M_2|\sum_{|s+n|\leq p^{1/2+\epsilon+\sigma}}\frac{p}{\pi |s|}\\
			&\leq p|M_2|\log(\frac{n+p/|M_2|}{n-p/|M_2|})\\
			&=p|M_2|\log(1+\frac{1}{n|M_2|/p-1})\\
			&\sim p^{2-1/2-\delta+\epsilon}=p^{3/2-\delta+\epsilon}.
\end{split}
\end{equation}

	For $|s|\geq p^{1/2+\epsilon+\sigma}\implies x_i\geq p^{\epsilon}$, we have
\[
	\sum_{x_i\geq p^{\epsilon}}O(p|M_1|x_i^{-2})+O(p^{3/2-\delta}x_i^{-1})+E(x_i)=O(p^{3/2-\sigma-\epsilon})+O(p^{3/2-\delta}\log p).
\]
	Thus, adding the two parts, we get
\[
	O(p^{3/2-\sigma-\epsilon})+O(p^{3/2-\delta}\log p)+O(p^{1/2-\delta+\epsilon})=O(p^{3/2-\alpha_\epsilon}),
\]
	where $\alpha_\epsilon=\min\{\epsilon+\sigma, \delta-\epsilon\}$. Now, since $\epsilon$ is arbitrary, we can optimize $\alpha$ to be $\sigma+(\delta-\sigma)/2$. 
	
	For different components of $p_1^l(s), p_2^l(s)$, the same arguments work verbatim by re-enumerate the $x_i$ and $y_j$-intervals, so the same estimate holds. Note that $n\sim p^{1/2+\delta}$ where $\delta\in(0,1/2)$, so $p-n\sim p$.
	
\end{proof}

%======================================%

\section{Error Terms}\label{error_term}

	In this section, we show that the contributions from $f_1, f_2, f_3, f_4, f_4'$ are all negligible. In increasing order of difficulty, we shall start with $f_4'$ and end with $f_1$. The remaining error terms can be summed trivially over $J_i^x$ and $\{i:i\geq p^\epsilon\}$, and the proof will be omitted.
.

\subsection{Estimates for $f_4'$}

	First, we note that
\begin{equation}
\label{f_4_prime_expansion}
\begin{split}
	\frac{4p^{2}ij}{3n\pi^2}f_4'(j)&=\frac{4p^{2}ij}{3n\pi^2}\bigg(\frac{p^3}{|M_2|^3}\frac{1}{\tilde{y}_j^3}-\frac{1}{j^3}\bigg)\\
		&=\frac{4p^2ij}{3n\pi^2}\bigg(\frac{3\frac{p^2}{|M_2|^2}j^2n-3\frac{p}{|M_2|}jn^2+n^3}{\tilde{y}_j^3 j^3}\big)\\
		&=O(p^4|M_2|^{-2}ij^{-1}\tilde{y}_j^{-3})+O(p^3|M_2|^{-1}nij^{-2}\tilde{y}_j^{-3})+O(p^2n^2ij^{-3}\tilde{y}_j^{-3}).
\end{split}
\end{equation}

\begin{lemma}
\label{f_4_prime_sum}
	For fixed integers $l,k>0$, one has
\[
	\sum_{j\in J_i^x}\frac{1}{j^l\tilde{y}_j^k}=O(\min_{0\leq s\leq k}\{(|M_2|^{-s}|M_1|^s i^{-s})\big(\frac{p}{|M_2|n}\big)^{l-s}\}p^{-k}i^{-k}|M_2|^{1}|M_1|^{k-1}),
\]
	where the constant depends on $l, k$. 
\end{lemma}	
		
	With Lemma \ref{f_4_prime_sum}, we can prove Proposition \ref{error_f_prop} (a).

\begin{proof} of Proposition \ref{error_f_prop} (a):

	From \eqref{f_4_prime_expansion}, we can use Lemma \ref{f_4_prime_sum}, choosing the parameter $s$ to be $0,1,2$ respectively for the three terms. Noting that $|M_1|\leq|M_2|$, we get the desired estimate bound.

\end{proof}
		
\begin{proof} of Lemma \ref{f_4_prime_sum}:

\[
\begin{split}
	\sum_{j\in J_i^x}\frac{1}{\tilde{y}_j^k}&\sim\frac{|M_2|}{p}(\frac{|M_1|}{p})^{k-1}\bigg[\frac{1}{i^{k-1}}-\frac{1}{(i+1)^{k-1}}\bigg]\\
			&=\frac{|M_2||M_1|^{k-1}}{p^k}\frac{(i+1)^{k-1}-i^{k-1}}{(i(i+1))^{k-1}}\\
			&\sim\frac{|M_2||M_1|^{k-1}}{x_i^kp^k},
\end{split}
\]	
	where we note that $x_{i+1}=x_i+1$. For the second equation, denoting $\frac{i|M_2|}{|M_1|}+\frac{|M_2|n}{p}$ by $\tilde{x}_i$, we have
\[
\begin{split}
	\sum_{j\in J_i^x}\frac{1}{j^k}&\sim\frac{(\tilde{x}_i+1)^{k-1}-\tilde{x}_i^{k-1}}{(\tilde{x}_i\tilde{x}_{i+1})^{k-1}}\\
			&\sim O(\min_{0\leq s\leq k}\{(|M_2|^{-s}|M_1|^s i^{-s})\big(\frac{p}{|M_2|n}\big)^{k-s}\}),
\end{split}
\]	
	where we note that 
\[\frac{1}{\tilde{x_i}}=O(\min\{\frac{|M_1|}{|M_2|i},|\frac{p}{|M_2|n}|\}).\]
	
	Now, by H\"{o}lder's inequality, we can derive the result.
	
\end{proof}

%-------------------------------------------------%

\subsection{Estimates for $f_3$ and  $f_4$}
	
	We are going to use the comparison lemma: If $f(x)$ is monotone, then
\[
	|\sum_{x=a}^b f(x)-\int_{a-1}^b f(t)\, dt|\leq |\int_{a-1}^b f(t)\, dt-\int_{a}^{b+1}f(t)\,dt|.
\]

\begin{lemma}
\label{f_estimate}
	The following statements are true:
\begin{itemize}
\item	$f_3(j)=O(|M_2|p^{-1}j^{-2})$,
\item	$f_4(j)=O(|M_2|^2p^{-2}j^{-3})$.
\end{itemize}
	The constant of the big-O notation is independent of $|M_2|$ and $p$.
\end{lemma}

\begin{proof}

	For $f_3$, we have that
\[
\begin{split}
	|f_3(j)|&\leq|\log(\frac{\tilde{y}_{j+1}}{\tilde{y}_j})-\log(\frac{\tilde{y}_{j+1}+1}{\tilde{y}_j+1})|\\
		&=|\log(\frac{\tilde{y}_{j+1}(\tilde{y}_{j+1}+1)}{\tilde{y}_j(\tilde{y}_{j+1}+1)})|\\
		&=|\log(1+\frac{\tilde{y}_{j+1}-\tilde{y}_j}{\tilde{y}_j(\tilde{y}_{j+1}+1)})|\\
		&=|\log(1+\frac{p}{|M_2|}\frac{1}{\frac{p^2 j^2}{|M_2|^2}-2\frac{pjn}{|M_2|}+n^2+\frac{pj}{|M_2|}(\frac{p}{|M_2|}+1)})|\\
		&=O(|M_2|p^{-1}j^{-2}).
\end{split}
\]

	For $f_4$, note that $\frac{1}{s(s+n)}=\frac{1}{n}(\frac{1}{s}-\frac{1}{s+n})$ is monotone.
\[
\begin{split}
	|f_4(y_j)|&\leq |\log(\frac{\tilde{y}_{j+1}(\tilde{y}_j+n)}{\tilde{y}_j(\tilde{y}_{j+1}+n)})-\log(\frac{(\tilde{y}_{j+1}+1)(\tilde{y}_j+n+1)}{(\tilde{y}_j+1)(\tilde{y}_{j+1}+n+1)})|\\
		&=|\log(1-\frac{1}{\tilde{y}_{j+1}+1})+\log(1-\frac{1}{\tilde{y}_{j}+(n+1)})-\log(1-\frac{1}{\tilde{y}_j+1})-\log(1-\frac{1}{\tilde{y}_{j+1}+(n+1)})|\\
		&=|(-\frac{1}{\tilde{y}_{j+1}+1}-\frac{1}{\tilde{y}_{j}+(n+1)}+\frac{1}{\tilde{y}_j+1}+\frac{1}{\tilde{y}_{j+1}+(n+1)})\\
		&+\frac{1}{2}(-\frac{1}{(\tilde{y}_{j=1}+1)^2}-\frac{1}{(\tilde{y}_j+(n+1))^2}+\frac{1}{(\tilde{y}_j+1)^2}+\frac{1}{(\tilde{y}_{j+1}+(n+1))^2})|+O(\frac{1}{\tilde{y}_j^3})\\
		&=|\bigg(\frac{p/|M_2|}{(\tilde{y}_{j+1}+1)(\tilde{y_j}+1)}-\frac{p/|M_2|}{(\tilde{y}_j+(n+1))(\tilde{y}_{j+1}+(n+1))}\bigg)\\
		&+\frac{1}{2}\bigg(\frac{2\frac{p}{|M_2|}(\frac{pj}{|M_2|}+1-n)+\frac{p^2}{|M_2|^2}}{(\tilde{y}_j+1)^2(\tilde{y}_{j+1}+1)^2}-\frac{1}{2}\bigg(\frac{2\frac{p}{|M_2|}(\frac{pj}{|M_2|}+1)+\frac{p^2}{|M_2|^2}}{(\tilde{y}_j+1+n)^2(\tilde{y}_{j+1}+1+n)^2}\bigg)|+O(\tilde{y}_j^3)\\
		&=|\frac{p}{|M_2|}\frac{n(\tilde{y_j}+\tilde{y}_{j+1})+(n+1)^2-1}{(\tilde{y}_{j+1}+1)(\tilde{y}_j+1)(\tilde{y}_j+(n+1))(\tilde{y}_{j+1}+(n+1))}|+O(|M_2|^2p^{-2}j^{-3})+O(|M_2|^3p^{-3}j^{-3})\\
		&=O(|M_2|^2p^{-2}j^{-3}).
\end{split}
\]
\end{proof}

\begin{proposition}
\[
	\sum_{i=-|M_1|/2}^{|M_1|/2}\sum_{j\in J_i^x}\frac{4p|M_2|i}{\pi^2}|f_3(j)|+\frac{4p^2ij}{n\pi^2}|f_4(j)|=O(|M_1||M_2|\log p).
\]
\end{proposition}

\begin{proof}
	Note that, by Lemma \ref{f_estimate}, 
\[
-\frac{4p|M_2|i}{\pi^2}f_3(j)+\frac{4p^2ij}{n\pi^2}f_4(j)=O(|M_2|^2ij^{-2})+O(|M_2|^2n^{-1}ij^{-2})=O(|M_2|^2ij^{-2}). 
\]

	Now, 
\begin{equation}
\begin{split}
	\sum_{j\in J_i^x}\frac{1}{j^2}
		&\sim\frac{1}{\frac{i|M_2|}{|M_1|}+\frac{|M_1|n}{p}}-\frac{1}{\frac{(i+1)|M_2|}{|M_1|}+\frac{|M_1|n}{p}}\\
		&=\frac{|M_2|/|M_1|}{(\frac{i|M_2|}{|M_1|}+\frac{|M_1|n}{p})(\frac{(i+1)|M_2|}{|M_1|}+\frac{|M_1|n}{p})}\\
		&=O(\frac{|M_1|}{|M_2|i^2}).
\end{split}
\end{equation}

	Then, summing over all possible $x_i$, we see that
\[
	\sum_{x=1}^{|M_1|}|M_1||M_2|\frac{x}{x^2}\sim|M_1||M_2|\log p,
\]
	which concludes the proof.
\end{proof}

%------------------------------------------%

\subsection{Estimates for $f_2$}

\begin{proposition}
\[
\sum_{i=-|M_1|/2}^{|M_1|/2}\sum_{j\in J_i^x}\frac{4p|M_1|j}{\pi^2}f_2(j)=O(|M_1||M_2|\log |M_2|).
\]

\end{proposition}

\begin{proof}	
	Suppose $\{\frac{p}{|M_2|}\}=\delta$, where $\{x\}=x-\lfloor x\rfloor$. Let $\{\frac{py_j}{|M_2|}\}=1-\epsilon=1-\epsilon_j$, then for a given $t\in\ZZ$,
\[
\begin{split}
	\frac{j}{t}-j\log(\frac{t+1-\epsilon}{t-\epsilon})&=\frac{j}{t}-y_j\log(1+\frac{1}{t-\epsilon})\\
					&=\frac{j}{t}-j(\frac{1}{t-\epsilon}-\frac{1}{2(t-\epsilon)^2}+O(t^{-3}))\\
					&=\frac{-\epsilon j}{t(t-\epsilon)}+\frac{j}{2(t-\epsilon)^2}+O(t^{-3}j).
\end{split}
\]
	Summing over $t$ from $\lceil\frac{pj}{|M_2|}\rceil$ to $\lfloor\frac{py(j+1)}{|M_2|}\rfloor$, we have
\[
\begin{split}
	\sum_{t=\lceil\frac{pj}{|M_2|}\rceil}^{\lfloor\frac{p(j+1)}{|M_2|}\rfloor}\frac{-j\epsilon}{t(t-\epsilon)}&=\sum_{t=\lceil\frac{pj}{|M_2|}\rceil}^{\lfloor\frac{p(j+1)}{|M_2|}\rfloor}y_j(\frac{1}{t}-\frac{1}{t-\epsilon})\\
	&\sim j\bigg(\log\frac{\lceil \frac{p(j+1)}{|M_2|}\rceil}{\lceil\frac{pj}{|M_2|}}-\log\frac{\frac{p(j+1)}{|M_2|}+(1-2\epsilon-\delta)}{\frac{pj}{|M_2|}}\bigg)\\
	&=j\bigg(\log\frac{(\frac{p(j+1)}{|M_2|}+(1-\epsilon-\delta))\frac{pj}{|M_2|}}{(\frac{pj}{|M_2|}+\epsilon)(\frac{p(j+1)}{|M_2|}+(1-2\epsilon-\delta))}\bigg)\\
	&=j\bigg(\log\bigg(1+\frac{\epsilon\frac{pj}{|M_2|}-\epsilon\frac{p(j+1)}{|M_2|}+O(1)}{(\frac{pj}{|M_2|}-\epsilon)(\frac{p(j+1)}{|M_2|}+(1-2\epsilon-\delta))}\bigg)\bigg)\\
	&=j\frac{\epsilon\frac{pj}{|M_2|}-\epsilon\frac{p(j+1)}{|M_2|}}{(\frac{pj}{|M_2|}+\epsilon)(\frac{p(j+1)}{|M_2|}+(1-2\epsilon-\delta))}+O(\frac{|M_2|^2}{p^2}j^{-1})\\
	&=\frac{-\epsilon\frac{pj}{|M_2|}}{(\frac{pj}{|M_2|}+\epsilon)(\frac{p(j+1)}{|M_2|}+(1-2\epsilon-\delta))}+O(\frac{|M_2|^2}{p^2}j^{-1})\\
	&=O(|M_2|p^{-1}j^{-1})+O(|M_2|^2p^{-2}j^{-1}).
\end{split}
\]
	The other term can be obtained similarly. Now, 
\[
	\sum_{j=1}^{|M_2|}\frac{4p|M_1|j}{\pi^2}f_2(j)=\sum_{j=1}^{|M_2|}O(|M_1||M_2|j^{-1})=O(|M_1||M_2|\log |M_2|).
\]
\end{proof}

%------------------------------------------%

\subsection{Estimates for $f_1$}

\begin{proposition}
	$\sum_{i\geq p^\epsilon}\sum_{j\in J_i^x}|M_1||M_2|f_1(j)=O(|M_2|^2)=O(p)$ if $|M_1|$ is even.
\end{proposition}

\begin{proof}

	Since $(p,|M_2|)=1$, we see that the fractional part of $\{py_j/|M_2|\}_{j=1}^{|M_2|}$ runs through $\{k/|M_2|\}_{k=0}^{|M_2|-1}$.
	
	We denote the fractional part of a number $x$ by $\{x\}=x-\lfloor x\rfloor$. Let $\{p/|M_2|\}=\delta$, and $\{py_j/|M_2|\}=\epsilon_j$, then 
\[
f_1(y_j)=-(1-\epsilon_j)-\{\epsilon_j+\delta\}=\left\{\begin{array}{lcl}
				-1-\delta&\text{if}& \epsilon+\delta<1\\
				-\delta&\text{if}&\epsilon+\delta\geq1
				\end{array}\right.
\]

	Without loss of generality, we may assume that $\delta\leq1/2$. Since $f_1$ changes signs from one $y_j$-interval to another, it is important to identify where $|f_1|$ attains $\delta$.

	In order to do that, we first introduce the notion of the critical zone.

\begin{definition}
	Given $\delta\leq 1/2$, the critical zone $\bar{A}\subset S^1$, the unit circle, is defined as $\bar{A}=[1-\delta,1)$. The discrete counterpart $A\subset \ZZ/|M_2|\ZZ$ is $A=\{x\in\ZZ/|M_2|\ZZ:\, \frac{x}{|M_2|}\in\bar{A}\}$.
\end{definition}

	We should note that $\{pj/|M_2|\}\in A$ if and only if $|f_1(j)|=\delta$. Thus, the problem now depends on when $\{pj/|M_2|\}$ lies in $A$ so as to account for cancellation.
	
	Now, we note that there are effectively $|M_2|/|M_1|$ $y_j$-intervals within one $x_i$-interval. Also, the corresponding $y_j$-intervals in consecutive $x_i$-intervals have different signs. In particular, $y_{j+2|M_2|/|M_1|}$-interval and $y_j$-interval have the same sign. Since we assume that $|M_1|$ is even, $\frac{2|M_2|}{|M_1|}\ZZ/|M_2|\ZZ$ is an additive subgroup of order $|M_1|/2$. Also, for any given $j$, $\{[\{py_k/|M_2|\}]\}_{k=j}^{j+2|M_2|/|M_1|-1}$ are distinct representatives of the coset.
	
	As $p$ is a unit in $\ZZ/|M_2|\ZZ$, we can replace the representatives by $\{-k\}_{k=1}^{|M_1|/2}$. Also, we see that between each coset, the number of elements inside the critical zone $A$ differs by at most $1$. Thus, the excessive parts that are not cancelled contribute at most $|M_2|/|M_1|$.
	
	For the boundary contribution of one $x_i$-interval, we see that the incomplete sums on both sides combine to represent the coset $|M_2|/|M_1|$.
	
	The argument above applies for summation over the whole group, but in our case we need to avoid the singularity at $-n$, which splits the summation range into 2 parts. Nonetheless, we shall show that the intuition still holds true even with segmented sums.

	If $\gamma=\{p/|M_2|\}<p^{-1/2+\sigma}$, then $|f_1(j)|=\delta$ for at most $p^\sigma$ times, so the contribution is $\sqrt{p}|M_2|=O(p)$.
	
	First, when we split the summation range into $2$ parts, note that since the complete summation gives at most the order of $|M_2|/|M_1|$, it suffices to estimate for one part and get the estimate of the other part by subtraction.
	
	As it suffices to estimate for the range $-n\leq s\leq p/2$, we are looking at the following quantity
\[
	I=\sum_{a\leq t\leq b}\sum_{j=0}^{\frac{2|M_2|}{|M_1|}-1}(-1)^{j}g\bigg[pj+\frac{tp|M_2|}{|M_1|}\bigg],
\]

	where $|b-a|=O(|M_1|)$, and $g=\mathbbm{1}_A:\ZZ/|M_2|\ZZ\to\RR$ is the characteristic function of $A$. Moreover, $|A|\sim\delta|M_2|$.
	
	Now,
\[
\begin{split}
	I&=\sum_{a\leq t\leq b}\sum_{j=0}^{\frac{2|M_2|}{|M_1|}-1}(-1)^{j}g[pj+\frac{tp|M_2|}{|M_1|}]\\
	&=\frac{1}{\sqrt{|M_2|}}\sum_{j}(-1)^{j}\sum_{a\leq t\leq b}\sum_{k\in\ZZ/|M_2|\ZZ}\hat{g}[k]e^{-2\pi\imath tpk/|M_1|}e^{2\pi\imath kpy_j/|M_2|}\\
	&=\frac{1}{\sqrt{|M_2|}}\sum_{k}\hat{g}[k]\bigg(\sum_{a\leq t\leq b}e^{-2\pi\imath tpk/|M_1|}\bigg)\bigg(\sum_{y_j}(-1)^{y_j}e^{2\pi\imath kpy_j/|M_2|}\bigg)\\
	&=\frac{1}{|M_2|}\sum_k\bar{C}_k\frac{\sin(\pi k|A|/|M_2|)}{\sin(\pi k/|M_2|)}\frac{\sin(\pi k(b-a+1)p/|M_1|)}{\sin(\pi pk/|M_1|)}\frac{\sin(2\pi kp/|M_1|)}{\sin(2\pi kp/|M_2|)}\sin(\pi kp/|M_2|),
\end{split}
\]
	where $|\bar{C}_k|=1$ for all $k$. Thus, by H\"{o}lder's inequality, the identity formula of the Fej\'{e}r kernel, and change of variables ($kp\mapsto l$), we see that
	
\[
\begin{split}
	|I|&\leq \frac{1}{|M_2|}\bigg(\sum_k|\frac{\sin(\pi k|A|/|M_2|)}{\sin(\pi k/|M_2|)}|^2\bigg)^{1/2}\cdot\\
	&\bigg(\sum_{l\in\ZZ/|M_2|\ZZ}\bigg|\frac{\sin(\pi l(b-a+1)/|M_1|)}{\sin(\pi l/|M_1|)}\bigg|^2\bigg|\frac{\sin(2\pi l/|M_1|)}{\sin(2\pi l/|M_2|)}\sin(\pi l/|M_1|)\bigg|^2\bigg)^{1/2}\\
	&\leq\frac{1}{|M_2|}\sqrt{|A|}\sqrt{(b-a+1)\frac{|M_2|}{|M_1|}}\frac{|M_2|}{|M_1|}\\
	&=O(\frac{|M_2|}{|M_1|}).
\end{split}
\] 
	As a result, the contribution from each ends is at most $|M_2|/|M_1|$, which concludes our proof.
\end{proof}

%================================%

%-------------------------------------%

\section{Acknowledgement}
The author gratefully acknowledges the support from ARO Grant W911NF-17-1-0014 and Dr. John Benedetto for the invaluable advice.

\bibliographystyle{amsplain}

\bibliography{ref}

\providecommand{\bysame}{\leavevmode\hbox to3em{\hrulefill}\thinspace}
\providecommand{\MR}{\relax\ifhmode\unskip\space\fi MR }
\providecommand{\MRhref}[2]{%
  \href{http://www.ams.org/mathscinet-getitem?mr=#1}{#2}
}
\providecommand{\href}[2]{#2}
\begin{thebibliography}{10}

\bibitem{AB_MF_DM_JM_2016}
Afonso~S Bandeira, Matthew Fickus, Dustin~G Mixon, and Joel Moreira,
  \emph{Derandomizing restricted isometries via the legendre symbol},
  Constructive Approximation \textbf{43} (2016), no.~3, 409--424.

\bibitem{JB_RB_JW_2012}
John~J Benedetto, Robert~L Benedetto, and Joseph~T Woodworth, \emph{Optimal
  ambiguity functions and weil's exponential sum bound}, Journal of Fourier
  Analysis and Applications \textbf{18} (2012), no.~3, 471--487.

\bibitem{JB_2011}
Jean Bourgain, Stephen Dilworth, Kevin Ford, Sergei Konyagin, Denka Kutzarova,
  et~al., \emph{Explicit constructions of rip matrices and related problems},
  Duke Mathematical Journal \textbf{159} (2011), no.~1, 145--185.

\bibitem{EC_JR_TT_2006_1}
Emmanuel~J Candes, Justin~K Romberg, and Terence Tao, \emph{Stable signal
  recovery from incomplete and inaccurate measurements}, Communications on pure
  and applied mathematics \textbf{59} (2006), no.~8, 1207--1223.

\bibitem{EC_TT_2005}
Emmanuel~J Candes and Terence Tao, \emph{Decoding by linear programming}, IEEE
  transactions on information theory \textbf{51} (2005), no.~12, 4203--4215.

\bibitem{EC_TT_2006}
\bysame, \emph{Near-optimal signal recovery from random projections: Universal
  encoding strategies?}, IEEE transactions on information theory \textbf{52}
  (2006), no.~12, 5406--5425.

\bibitem{FC_1994}
Fan~RK Chung, \emph{Several generalizations of weil sums}, Journal of Number
  Theory \textbf{49} (1994), no.~1, 95--106.

\bibitem{JF_HI_1993}
John Friedlander and Henryk Iwaniec, \emph{Estimates for character sums},
  Proceedings of the American Mathematical Society \textbf{119} (1993), no.~2,
  365--372.

\bibitem{DM_2015}
Dustin~G Mixon, \emph{Explicit matrices with the restricted isometry property:
  Breaking the square-root bottleneck}, Compressed sensing and its applications
  (2015), 389--417.

\bibitem{WS_2006}
Wolfgang~M Schmidt, \emph{Equations over finite fields: an elementary
  approach}, vol. 536, Springer, 2006.

\bibitem{AW_1948}
Andr{\'e} Weil, \emph{On some exponential sums}, Proceedings of the National
  Academy of Sciences \textbf{34} (1948), no.~5, 204--207.

\end{thebibliography}

\end{document}